%% file: main.tex
\documentclass[12pt]{article}
\usepackage[margin=1in]{geometry} 
\usepackage[dvipsnames]{xcolor}
\usepackage{amsmath,amsthm,amssymb, dsfont, mathtools, amsfonts}
\usepackage{comment} 
\usepackage{tikz}
\usetikzlibrary{positioning}
\usepackage{graphicx}
\usepackage{enumitem} 
\usepackage{lipsum}
\usepackage[UKenglish]{isodate}
\usepackage{hyperref}
\hypersetup{
    colorlinks = true,
    citecolor=blue
}

\usepackage{setspace}
\usepackage[
backend=biber,
style=authoryear,
maxbibnames=20, 
uniquelist=false, 
uniquename=false, 
doi = false,
url = false,
eprint=true,
isbn=true
]{biblatex}
    \renewbibmacro{in:}{%
        \ifentrytype{article}{}{\printtext{\bibstring{in}\intitlepunct}}}
    \DefineBibliographyExtras{english}{}

\newtheorem{assumption}{Assumption}
\numberwithin{assumption}{section}
\newtheorem{proposition}{Proposition}
\numberwithin{proposition}{section}
\newtheorem{corollary}{Corollary}
\numberwithin{corollary}{section}
\newtheorem{remark}{Remark}
\usepackage{soul}

    \newtheorem{propositionalt}{Proposition}[proposition]
        \newenvironment{propositionp}[1]{
        \renewcommand\thepropositionalt{#1}
        \propositionalt
    }{\endpropositionalt}

\newcommand{\E}{\mathbb{E}} 
\newcommand{\Pb}{\mathbb{P}} 
\newcommand{\T}{\mathcal{T}}
\newcommand{\G}{\mathcal{G}}
\newcommand{\eq}{=} 

\definecolor{chocolate}{rgb}{1,0.5,0.14}
\definecolor{burgundy}{rgb}{0.78125, 0.1, 0.246}

\title{{Dynamic LATEs with a Static Instrument}\thanks{We would like to thank Vítor Possebom, Pedro Sant'Anna, Claudia Noack, Peter Hull, David Slichter, Sílvia Gonçalves, Paul Goldsmith-Pinkham, Lucas Lima, Camila Galindo, Carolina Caetano, Marinho Bertanha, and Joshua Angrist for excellent comments and suggestions. We also thank participants at the FGV-EESP seminar, and at the CEA 2023, the Causal Data Science Meeting, and the SBE 2023 conferences. This paper supersedes an earlier draft, titled ``Identifying Dynamic LATEs with a Static Instrument'' (\cite{ferman2023dLATEs}).}}
\author{Bruno Ferman\thanks{Sao Paulo School of Economics - FGV; email: bruno.ferman@fgv.br} \and Otávio Tecchio\thanks{MIT Department of Economics; email: otecchio@mit.edu}}
\date{\today}

\begin{document}
\maketitle


\begin{abstract}

In many situations, researchers are interested in identifying dynamic effects of an irreversible treatment with a time-invariant binary instrumental variable (IV). For example, in evaluations of dynamic effects of training programs with a single lottery determining eligibility. A common approach in these situations is to report per-period IV estimates. Under a dynamic extension of standard IV assumptions, we show that such IV estimands identify a weighted sum of treatment effects for different latent groups and treatment exposures. However, there is possibility of negative weights. We discuss point and partial identification of dynamic treatment effects in this setting under different sets of assumptions.

\hfill \break
    
\hfill \break
\end{abstract}

\newpage
\section{Introduction}

In many situations, researchers are interested in identifying dynamic effects of an irreversible treatment with a time-invariant binary instrumental variable (IV). As an example, consider evaluations of dynamic effects of training programs exploiting a single lottery determining eligibility for a given cohort (e.g., \cite{jobcorps, alzua2016, hirshleifer2016, das2021}). Another example is the estimation of the dynamic effects of fertility on labor market outcomes using exogenous variations such as twins at first birth, sex composition of the first two children, and in-vitro fertilization success (e.g., \cite{bronars_aertwins, angelov2012mothers, silles, lundborg2017}). A common approach in these situations is to report per-period reduced form (RF) or IV estimates using an any-exposure indicator  as the treatment variable.

We show that if observations can access treatment at any period, those common approaches may recover weighted sums of causal effects in which some weights are negative.  If first stages are decreasing over time, then there must be negative weights (and there may also be negative weights when the first stage is nondecreasing). We then extend the identification results by \textcite{ischemia}. Specifically, it is possible to identify dynamic local average treatment effects (LATEs) even when there are defiers after the first period under a generalization of their wave ignorability assumption. Finally, we consider partial identification of dynamic LATEs without requiring any restriction on treatment effect heterogeneity.

This paper is related to a few different strands of the econometrics and applied econometrics literature. \textcite{lundborg2017}  recognize the shortcoming of per-period IV estimands when estimating dynamic effects of fertility on women's labor market outcomes. However, they do not provide a formal decomposition in a general setting with heterogeneous treatment effects nor discuss point and partial identification. \textcite{miquel2002iv} considers identification of dynamic treatment effects with a static instrument under conditions  that   are unreasonable for applications such as estimating dynamic effects of training programs or fertility.\footnote{\textcite{miquel2002iv} assumes that potential outcomes are independent of the instrument conditional on a history of treatment assignments. However, in the context of training programs or fertility, conditioning on a history of realized treatments implies conditioning on different latent groups depending on whether $Z_i=1$ or $Z_i=0$.}

Our setting is also related to the literature on multi-valued treatments and lower dimensional instrumental variables (e.g., \cite{ai95,agi_fish,torgovitsky2015, dhault_lowZ, masten_torgovsitsky, caetano_escanciano_2021, hull2018isolateing}) and to the literature on fuzzy and instrumented difference-in-differences (e.g., \cite{chaisemartinFDD,hudsonDID, picchetti2022marginal}). Differently from the former, the dynamic structure of our setting allows for alterative identification results exploiting recursiveness. Differently from the latter, we do not explore time variation under parallel trend assumptions. 

Finally, our negative weights result is inserted in the recent developments on two-way fixed effects estimands (\cite{dechaisemartin2020_did, callaway2021,sun2021, GOODMANBACON2021254, ATHEY202262, borusyak2021revisiting}) and IV estimands with covariates (\cite{kolesar2013, blandholetal2022, sloczynski2022}). However, the drivers of negative weights in our setting are different. The recursive solution we discuss mostly resembles \textcite{cellinirdd}'s result on identification of dynamic effects in regression discontinuity designs. However, they only consider the case of regression discontinuity designs that are sharp and focus on a different set of target parameters.

This paper is organized as follows. Section \ref{Sec:2 periods} derives results for two periods, illustrating the principles at work. This includes decomposition results for the RF and IV estimands (Section \ref{Sec:2 periods decomposition}), point identification results (Section \ref{Sec:2 periods point id}), and partial identification results (Section \ref{Sec:2 periods partial id}). Section \ref{Sec:T periods} considers the general multi-period setting. Section \ref{sec:conclusion} provides concluding remarks. Proofs are gathered in the Appendix.

\section{Two-period setting} 
\label{Sec:2 periods}

A setting with two periods illustrates main ideas. Observations are indexed by $i$ and time is indexed by $t \in \{1,2\}$. We are interested in identifying dynamic effects of a binary treatment $D_{i,t}$ on some outcome $Y_{i,t}$. No unit is treated before the first period. There is selection into treatment, but we observe a time-invariant binary instrument $Z_i$.

Treatment is irreversible: once an observation is treated, it will be treated for all following periods. This is a common assumption in the difference-in-differences literature, and is known as staggered treatment adoption (e.g., \cite{callaway2021, sun2021, ATHEY202262, borusyak2021revisiting}). 

\begin{assumption}[Irreversible Treatment]
\label{irreverse}
$D_{i,1}=1\implies D_{i,2}=1$ almost surely (a.s.). 
\end{assumption}

Because treatment is irreversible, any possible sequence of treatment statuses at time $t$ can be identified by zero if the observation has never been treated and by $(1,\tau)$ if the observation's first period of treatment was $t-\tau$. At $t=1$ observations may have treatment status $0$ (not treated at $t=1$) or $(1,0)$ (treated at $t=1$). In this case,  $\tau=0$ indicates that treatment length is zero, because the treatment started at $t=1$, and we are considering the observation at $t=1$. At $t=2$, in addition to treatment status 0, we may have $(1,1)$ (treated at $t=1$, so $\tau=1$ means that at $t=2$ the length of the treatment is 1) or $(1,0)$ (treated at $t=2$).

Let $Y_{i,t}(0,z)$ denote the potential outcome when observation $i$ is not treated at $t$ and was instrument assigned to $z$, while $Y_{i,t}(1,\tau,z)$ is the potential outcome when $i$ is first treated at $t-\tau$ and assigned by the instrument to $z$. Potential treatment statuses at $t$ are denoted by $D_{i,t}(z)$. Also, $AT_t$ denotes always-takers at $t$ (observations such that $D_{i,t}(1)=D_{i,t}(0)=1$), $C_t$ denotes compliers at $t$ (observations such that $D_{i,t}(1)>D_{i,t}(0)$), $F_t$  denotes defiers at $t$ (observations such that $D_{i,t}(1)<D_{i,t}(0)$) and $NT_t$ denotes never-takers at $t$ (observations such that $D_{i,t}(1)=D_{i,t}(0)=0$). 

In principle, there could be 16 latent groups, which are combinations of $(AT_t,C_t,F_t, NT_t)$ for the two periods. However, Assumption \ref{irreverse} restricts these possibilities. In particular, the group $AT_1$ must also be $AT_2$. Moreover, the group $C_1$ must be either $AT_2$ (in case those with $Z_i=0$ become treated in the second period) or $C_2$ (in case they remain untreated in the second period). In contrast, the group $NT_1$ can be any of the four possible latent groups in the second period even when treatment is irreversible. We say compliance is dynamic when there exist observations whose latent groups change over time.  Otherwise,  compliance is defined as static. Compliance is static if, for example, treatment is only accessed in the first period.

For each $t\in \{1,2\}$, define
\begin{equation}
\label{reducedform}
    RF_t\coloneqq
    \E\left[Y_{i,t}\middle|Z_i=1\right]
    -\E\left[Y_{i,t}\middle|Z_i=0\right]
\end{equation} 
and 
\begin{equation}
\label{firststage}
    FS_t\coloneqq
    \E\left[D_{i,t}\middle|Z_i=1\right]
    -\E\left[D_{i,t}\middle|Z_i=0\right],
\end{equation}
the per-period reduced form and first stage estimands at $t$, respectively. Thus, whenever $FS_t\neq0$, the per-period IV estimand at $t$ is $RF_t/FS_t$.

As a first requirement for $Z_i$ to be considered a valid instrument, we consider a dynamic extension of the standard IV assumptions of \textcite{ai94} and \textcite{air96}. The main difference from the assumptions in the static case is that we add independence and exclusion conditions in all periods. Note that relevance and monotonicity assumptions are only required in the first period.

\begin{assumption}
\label{basicvalidity}
The following hold:
\begin{enumerate}
    \item \textit{Exclusion}: For each  $z\in\{0,1\}$, $Y_{i,t}(0, z)=Y_{i,t}(0)$ and $Y_{i,t}(1,0, z)=Y_{i,t}(1,0)$ for $t \in \{1,2\}$, and $Y_{i,2}(1,1, z)=Y_{i,2}(1,1)$.
    
    \item \textit{Independence}: $\big(Y_{i,1}(0), Y_{i,1}(1,0), Y_{i,2}(0), Y_{i,2}(1,0),Y_{i,2}(1,1), D_{i,1}(1), D_{i,1}(0),D_{i,2}(1), D_{i,2}(0)\big)$ is independent of $Z_i$.
    \item \textit{Relevance at $t=1$}: $FS_1\neq0$.
    \item \textit{Monotonicity at $t=1$}: $\mathbb{P}(F_1)=0$.
\end{enumerate}
\end{assumption}

Our focus will be on comparisons between treated and untreated potential outcomes. Thus, the building blocks for decomposing the per-period reduced form estimands are causal effects of the form\footnote{Whenever written, expectations are assumed to exist.}
\begin{equation}
\label{treatmenteffects}
    \Delta_t^\tau(g)\coloneqq
    \E\left[Y_{it}(1, \tau)-Y_{it}(0)\mid g\right],
\end{equation}
where $g$ specifies a history of IV latent types. For example, an observation that is only treated in the first period if $Z_i=1$ but, in the second period, gets treated regardless of $Z_i$ belongs to $g=(C_1, AT_2)$. In this case, $ \Delta_2^0(C_1, AT_2)$ is the treatment effect for this group of observations at $t=2$ when they receive treatment at $t=2$.  Note that there are three types of time heterogeneity in these treatment effects. The first one is with respect to the calendar time $t$, the second one is with respect to the treatment length $\tau$, while the third one is with respect to the latent group. 

We focus on target parameters of the type $\Delta_t^{t-1}(C_1)$, which we {term} ``dynamic LATEs''. These are the local average treatment effects at time $t$, when treatment started at $t=1$, for first-period compliers ($C_1$). For the comparison of effects across time to be valid, it is important that the IV latent type for which the causal effect is identified does not change. On the contrary, differences in effects across time cannot be solely attributed to time heterogeneity.  

Given the notation above, it follows directly from \textcite{ai94} that $\Delta^0_1(C_1)$ is identified by the first period IV estimand under Assumption \ref{basicvalidity}. Moreover, in case of static compliance, Assumptions \ref{irreverse} and \ref{basicvalidity} imply that the IV estimand in the second period identifies $\Delta^1_2(C_1)$, the effect at $t=2$ of being treated at $t=1$ for $C_1$ observations. The argument for identification is analogous to the one for the first period.

\subsection{Decomposition of RF and IV estimands} 
\label{Sec:2 periods decomposition}

While, under Assumptions \ref{irreverse} and \ref{basicvalidity},  the  IV estimands recover the dynamic LATEs when there is static compliance, the second-period IV estimand generally does not recover  $\Delta_2^{1}(C_1)$  when there is dynamic compliance. 

Figure \ref{figsecondperiod} depicts the remaining latent groups at $t=2$ once latent groups not consistent with irreversible treatment and first-period defiers are excluded (Assumptions \ref{irreverse} and \ref{basicvalidity}).  It is clear that the averages for $g = (AT_1,AT_2)$ cancel out in $RF_2 = \E[Y_{i,2}|Z_i=1] - \E[Y_{i,2}|Z_i=0]$ because the observed outcomes for them are the same potential outcomes regardless of $Z_i$. The same is true for $g=(NT_1,AT_2)$ and $g=(NT_1,NT_2)$.

\input{figure}

Therefore, $RF_2$ captures the comparisons for remaining latent groups. The main problem, however, is that for some of those groups the difference in observed outcomes between those with $Z_i=1$ and $Z_i=0$ does not represent a difference between potential outcomes $Y_{i,2}(1,1)$ and $Y_{i,2}(0)$. In particular,
\begin{equation*}
\label{contC1AT2}
    \E[Y_{i,2} | Z_i=1,C_1] - \E[Y_{i,2} | Z_i=0,C_1] = \Delta_2^1(C_1) - \Pb(AT_2 \mid C_1) \Delta_2^0(C_1,AT_2).
\end{equation*}
Moreover, the differences in expected outcomes for the groups $(NT_1,C_2)$ and $(NT_1,F_2)$ equal a causal effect of treatment length zero. The following proposition characterizes the $RF_2$ and $FS_2$ estimands when there is dynamic compliance.

\begin{proposition}
\label{decomposition2}
    Under Assumptions \ref{irreverse} and \ref{basicvalidity}, 
    \begin{equation}
    \label{decompRF2}
    \begin{split}
    RF_2&=
    \Pb(C_1)\Delta^1_2(C_1)\\
    &-\Pb(C_1,AT_2)
    \Delta^0_2(C_1,AT_2)
    -\Pb(NT_1,F_2)
    \Delta^0_2(NT_1,F_2)\\
    &+\Pb(NT_1,C_2)
    \Delta^0_2(NT_1,C_2)
    \end{split}
    \end{equation} 
    and
    \begin{equation}
    \label{decompFS2}
    FS_2=\Pb(C_1)-\Pb(C_1,AT_2)
        -\Pb(NT_1,F_2)+\Pb(NT_1,C_2).
    \end{equation}
\end{proposition}
\begin{proof}
Special case of Proposition \ref{rffsestimands}.
\end{proof}

Equation \eqref{decompRF2} shows that $RF_2$ depends on the dynamic LATE of interest at $t=2$, $\Delta^1_2(C_1)$, but also on the effects for some groups that switch into treatment in the second period. In particular, because the $(C_1, AT_2)$ and $(NT_1, F_2)$ get treated at $t=2$ only when $Z_i=0$, the causal effect for them is negatively weighted. A negative weight for the $(C_1, AT_2)$ group is specially relevant because it implies that assuming no defiers in all periods is not sufficient to avoid negative weights. In fact, the decomposition for the $FS_2$ in Equation \eqref{decompFS2} shows that whenever $FS_2<FS_1=\Pb(C_1)$, there must be negative weights in $RF_2$ regardless of assumptions  on the existence of specific latent groups. More generally, for settings with $T$ periods,  Corollary \ref{IVestimand} shows that if there is a period in which the first stage is strictly smaller than in the period before, then there must be negative weights in the reduced form of current and future periods.

Equation \eqref{decompRF2} also indicates a typical case in which there might be sign reversal in the sense that all causal effects have the opposite sign of $RF_2$. Ignoring the $NT_1$'s in $RF_2$ for the sake of the argument, if effects fade out sufficiently fast with respect to the treatment length dimension, then the term related to $(C_1, AT_2)$ in $RF_2$ could be larger than the term related to $C_1$. For example, for the effects of children on parents' labor supply the treatment length dimension is the age of the child. Thus, if effects are always negative but decrease (in absolute value) when children get older, the reduced form estimand could be positive.

Given this decomposition for the reduced form and for the first stage, the decomposition for the IV estimand at $t=2$ is immediate. Corollary \ref{IVestimand2} summarizes its main characteristics. The two main takeaways are that negative weights in $RF_2$ imply negative weights in the IV estimand and that the weights in the IV estimand sum to one.

\begin{corollary}
\label{IVestimand2}
Under Assumptions \ref{irreverse} and \ref{basicvalidity}, if $FS_2\neq0$, $RF_2/FS_2$ is a linear combination  of the causal effects in Equation \eqref{decompRF2} in which the weights sum to one but some of them may be negative. There must be negative weights whenever $FS_2<FS_1$. Moreover, the causal effects that are negatively weighted in $RF_2/FS_2$ are the same as in $RF_2$ if, and only if, $FS_2>0$. 
\end{corollary}
\begin{proof}
    Special case of Corollary \ref{IVestimand}.
\end{proof}

Given the results above, it is straightforward to consider assumptions under which the second period IV estimand recovers $\Delta^1_2(C_1)$. One case is when compliance is static. In this case, observations do not change treatment status from the first period to the second, implying
\begin{equation*}
\Pb(C_1, AT_2)=\Pb(NT_1, C_2)=\Pb(NT_1, F_2)=0,
\end{equation*}
and so $RF_2$ reduces to $\Pb(C_1)\Delta^1_2(C_1)$ while $FS_2=\Pb(C_1)$.  
However, this is not the only case in which the IV estimand works. Assumption \ref{Assumption_static_IV} formalizes types of treatment effects homogeneities which guarantee that the IV estimand at $t=2$ identifies a causal effect.

\begin{assumption} \label{Assumption_static_IV}
For any latent group $g \in \left\{(C_1,AT_2),(NT_1,C_2),(NT_1,F_2)\right\}$ such that $\Pb(g)>0$, $\Delta_2^1(C_1) = \Delta_2^0(g)$.
    
\end{assumption}

\begin{corollary}
\label{Ivestimand:tauhomo}
Suppose Assumptions \ref{irreverse} and \ref{basicvalidity} hold. Under Assumption \ref{Assumption_static_IV}, and if $FS_2\neq0$, 
\begin{equation*}
    \Delta^1_2(C_1)=\frac{RF_2}{FS_2}.
\end{equation*}
\end{corollary}
\begin{proof}
    This result is immediate given Proposition \ref{decomposition2}.
\end{proof}

Assumption \ref{Assumption_static_IV} is trivially satisfied if treatment effects are fully homogeneous (that is, with respect to treatment length, calendar time, and latent group). More generally, it says that for groups contaminating $RF_2$, average treatment effects at $t=2$ must be the same as the LATE at $t=2$ for first-period compliers (who were treated at $t=1$). This condition encompasses two sources of treatment effects homogeneity. First, it requires that treatment effects do not depend on the time since those observations have been treated. This condition is arguably too strong in many settings. For example, as already discussed,  effects of fertility on labor supply {are most likely stronger} when the treatment length is smaller. Likewise,  training programs {likely} have negative effects in the beginning (while subjects are still taking classes), and then positive effects afterward. Second, Assumption \ref{Assumption_static_IV} requires  treatment effects  for latent groups that contaminate $RF_2$  to be the same as for first-period compliers. On the other hand, note that Assumption \ref{Assumption_static_IV} does not impose restrictions on the possibility that treatment effects vary with calendar time. Corollary \ref{Assumption_static_IV} is analogous to Theorem 3 by \textcite{ischemia}.

\begin{remark}
    Defining potential outcomes as $\widetilde Y_{i,t}(1,z)$ when observation $i$ is treated in the initial period and $\widetilde Y_{i,t}(0,z)$ otherwise would not be a valid solution without further assumptions. In this case, $\widetilde Y_{i,t}(0,z)$ would depend on $z$ if compliance was dynamic, so the usual IV exclusion restriction would not be valid for this definition of potential outcomes. For example, the instrument directly affects the potential outcome $\widetilde Y_{i,2}(0,z)$ for $(NT_1, C_2)$ observations because they are treated at $t=2$ only when $Z_i=1$.
\end{remark}

\subsection{Point identification of dynamic LATEs}  
\label{Sec:2 periods point id}

Dynamic LATEs can be identified without restricting heterogeneity with respect to the treatment length dimension. This comes at the cost of imposing homogeneity with respect to calendar time. Assumption \ref{Assumption_id} formalizes this alternative homogeneity assumption.

\begin{assumption} 
\label{Assumption_id}
For any latent group $g \in \left\{(C_1,AT_2),(NT_1,C_2),(NT_1,F_2)\right\}$ such that $\Pb(g)>0$, $\Delta_1^0(C_1) = \Delta_2^0(g)$. 
\end{assumption}

Assumption \ref{Assumption_id} says that for groups contaminating $RF_2$, average treatment effects at $t=2$ must be the same as the first-period LATE. The main difference from Assumption \ref{Assumption_static_IV} is  {the change in} the type of time heterogeneity. To understand the economic difference of these assumptions, it is useful to go back to the training program case. If, for example, the outcome of interest is employment, then causal effects most likely depend on whether the economy is in a recession or in a boom phase. Thus, homogeneity with respect to calendar time would be a strong assumption in a period of strong economic fluctuations. On the other hand, in periods of economic stability, it could be reasonable to assume that effects do not depend on calendar time. Therefore, at least when {the economy is stable}, Assumption \ref{Assumption_id} should be more palatable than Assumption \ref{Assumption_static_IV} in these applications.

The existence of latent groups $(NT_1,C_2)$ and $(NT_1,F_2)$ depends crucially on the empirical setting. Once more, consider the training program example. Suppose first that  being lottery assigned to treatment implies that admission is guaranteed not only in the current period, but also in the following ones. In this case, some of the $NT_1$ observations might get treated in the second period only when they have a guaranteed admission (in this case, when they have $Z_i=1$). Therefore, we should expect $\Pb(NT_1,C_2)>0$. It is also conceivable to have empirical applications in which there are second-period defiers, even when {there are}  no first-period defiers. For example, imagine a setting in which those lottery assigned to treatment that refuse training in the first period cannot be trained in the second period. In that case, all first-period never-takers with $Z_i=1$ would not be trained in the second period, but some with $Z_i=0$ might. In this case, we would expect $\Pb(NT_1,F_2)>0$.

Alternatively, suppose the lottery in the initial period does not guarantee admission in the following periods, and that first-period never-takers do not receive different information depending on their $Z_i$. In this case, it would be more reasonable to assume that second-period take-up for $NT_1$ does not depend on instrument assignment, so $\Pb(NT_1,C_2)=\Pb(NT_1,F_2)=0$. Therefore, in these settings,  $\Delta_1^0(C_1) = \Delta_2^0(C_1,AT_2)$ suffices for identification. The same is true for settings with no $NT_1$ observations, which is the case when all observations are treated in the first period when $Z_i=1$.

Since $\Delta^0_1(C_1)$ is identified, it is possible to identify the contamination term of the reduced form estimand under Assumption \ref{Assumption_id}, and identify $\Delta^1_2(C_1)$ by correcting for the bias in $RF_2$. 

\begin{proposition}
\label{identification2}
    Suppose Assumptions \ref{irreverse} and \ref{basicvalidity} hold. Under Assumption \ref{Assumption_id}, 
    \begin{equation}
    \label{solvedelta2}
    \Delta^1_2(C_1)
    =\frac{RF_2}{FS_1}
    +\frac{\big(FS_1-FS_2\big)}{FS_1}
    \frac{RF_1}{FS_1}.
    \end{equation}
\end{proposition}
\begin{proof}
    Special case of Proposition \ref{identification}.
\end{proof}

Therefore, Proposition \ref{identification2} provides an alternative way to identify dynamic LATEs that (relative to the per-period IV estimator) relies on  more reasonable assumptions in many settings. Moreover, in contrast to the per-period IV estimand for $t=2$, the identification result in Proposition \ref{identification2} requires relevance only in the first period (that is, it could be that $FS_2=0$).

\textcite{ischemia} use wave ignorability to identify average exposure effects in ISCHEMIA. Proposition \ref{identification2} extends this to settings with defiers after the first period. The cost is requiring an additional treatment effect homogeneity in case $\mathbb{P}(NT_1,F_2)>0$. The recursive correction in \eqref{solvedelta2} can be automated by the linear two-stage least squares regression considered in \textcite{ischemia}'s Theorem 2.

\begin{remark}
Given the decomposition results from Proposition \ref{decomposition2}, it is possible to adapt the solution we propose in this section to other settings in which  more information is available. For example, suppose {there is a second}  lottery at $t=2$ that is independent from the first-period lottery, and let $\widetilde C_2$ be the compliers of this second lottery.\footnote{Observations who participated in the first-period lottery may self select into participating in the second-period lottery. Moreover, lottery participants in this second-period lottery may also include observations who did not participate in the first-period lottery.} In this case, $\Delta_2^0(\widetilde C_2)$ is identified. Therefore,{it can be used to} correct the contamination term  (instead of $\Delta_1^0( C_1)$) assuming that, for any latent group $g \in \left\{(C_1,AT_2),(NT_1,C_2),(NT_1,F_2)\right\}$ such that $\Pb(g)>0$, $\Delta_2^0(\widetilde C_2) = \Delta_2^0(g)$ (instead of Assumption \ref{Assumption_id}). In this case, {heterogeneity with respect to $t$ and $\tau$ is unrestricted}, but {there still are cross-group homogeneity restrictions}.
\end{remark}

\begin{remark}
Our framework can be extended to analyses of the causal effects of charter schools (\cite{AADKP2011,DF2011,GCD2011,ACDPW2016,AAHP2016}). For example, define potential outcome $Y_{i,t}(s,\tilde t)$ for a student $i$ at time $t$ were he/she enrolled in a charter school for the first time at time $\tilde t$ in grade $s$. Then we can define causal effects based on comparisons between $Y_{i,t}(s,\tilde t)$ and $Y_{i,t}(0)$, which is the potential outcome had the  student  never enrolled in a charter school until period $t$.\footnote{Note that the way $Y_{i,t}(s,\tilde t)$ is defined does not impose restrictions on the exposure to charter schools after initial enrollment. In this case, the number of years enrolled in a charter school is one of the mechanisms in which the treatment (in this case, being enrolled in a charter school for the first time at time $\tilde t$ in grade $s$) may affect outcomes. In the same way as college enrollment would be a mechanism in which charter school enrollment may affect earnings. An alternative  in this case would be to define potential outcomes as a function of the number of years (or the specific years) in a charter school. Appendix A from \textcite{AAHP2016} presents the interpretation of the IV estimand when the treatment variable is the number of years enrolled in a charter school ($\tilde d$), and potential outcomes are defined as a function of $\tilde d$.} When considering a lottery at $t=1$, we should take into account the possibility that students enroll in a charter school in subsequent periods, and our results can be adapted to this setting.
\end{remark}

\subsection{Partial identification of dynamic LATEs}
\label{Sec:2 periods partial id}

Dynamics LATEs are partially identified without any restriction on the treatment effect heterogeneity when treatment effects are bounded. Bounds for treatment effects are natural in, for example, settings with bounded outcomes (if there exist $\underline{Y}, \overline{Y}\in\mathbb{R}$ such that $\underline{Y}\le Y_{i,2}\le\overline{Y}$ with probability one, then the treatment effects are bounded, in absolute value, by $\overline{Y}-\underline{Y}$).

\begin{proposition}
\label{tebounds2}
Suppose Assumptions \ref{irreverse} and \ref{basicvalidity} hold. If there exist $\underline{\Delta}, \overline{\Delta}\in\mathbb{R}$, with $\underline{\Delta}\le0\le\overline{\Delta}$, such that for all $g\in\left\{(C_1,AT_2),(NT_1,C_2),(NT_1,F_2)\right\}$ with $\Pb(g)>0$, $\underline{\Delta}\le\Delta^0_2(g)\le\overline{\Delta}$, then a lower bound for $\Delta^1_2(C_1)$ is given by
\begin{equation}
\label{generallowerbound2}
\begin{split}
    \frac{RF_2}{FS_1}
    +
    \Pb\left(D_{i,2}>D_{i,1}\middle|Z_i=0\right)
    \frac{\underline{\Delta}}{FS_1}
    -
    \Pb\left(D_{i,2}>D_{i,1}\middle|Z_i=1\right)
    \frac{\overline{\Delta}}{FS_1}
\end{split}
\end{equation}
and an upper bound is given by
\begin{equation}
\label{generalupperbound2}
\begin{split}
    \frac{RF_2}{FS_1}
    +
    \Pb\left(D_{i,2}>D_{i,1}\middle|Z_i=0\right)
    \frac{\overline{\Delta}}{FS_1}
    -
    \Pb\left(D_{i,2}>D_{i,1}\middle|Z_i=1\right)
    \frac{\underline{\Delta}}{FS_1}.
\end{split}
\end{equation} 
If, in addition to the conditions above, for all $g,g'\in\left\{(C_1,AT_2),(NT_1,C_2),(NT_1,F_2)\right\}$ with $\Pb(g)>0$ and $\Pb(g')>0$, $\Delta^0_2(g)=\Delta^0_2(g')$, then 
\begin{equation}
\label{specificlowerbound2}
\frac{RF_2}{FS_1} + 
\Bigg[\mathbf{1}(FS_2\le FS_1)\underline{\Delta}+\mathbf{1}(FS_2> FS_1)\overline{\Delta}\Bigg]\frac{FS_1-FS_2}{FS_1},
\end{equation}
where $\mathbf{1}(\cdot)$ is the indicator function, is a lower bound for $\Delta^1_2(C_1)$ and 
\begin{equation}
\label{specificupperbound2}
\frac{RF_2}{FS_1} 
+\Bigg[\mathbf{1}(FS_2\le FS_1)\overline{\Delta}  + \mathbf{1}(FS_2> FS_1)\underline{\Delta}\Bigg]\frac{FS_1-FS_2}{FS_1}
\end{equation}
is an upper bound. These bounds are (weakly) tighter than the previous ones. 
\end{proposition}

\begin{proof}
    Special case of Proposition \ref{tebounds_T}.
\end{proof}

\begin{remark}
\label{remark_nonexistentgroups}
    Assuming $\Pb(NT_1, C_2)=\Pb(NT_1, F_2)=0$ implies that the conditions in Proposition \ref{tebounds2} for tighter bounds (Equations \eqref{specificlowerbound2} and \eqref{specificupperbound2}) hold. Section \ref{Sec:2 periods point id} discussed settings in which assuming $\Pb(NT_1, C_2)=\Pb(NT_1, F_2)=0$ should be reasonable. In those cases, the tighter bounds hold without any assumption on treatment effect heterogeneity. Moreover, $\Pb(NT_1, C_2)=\Pb(NT_1, F_2)=0$ also implies $FS_2\le FS_1$, so that
    \begin{equation*}
    \frac{RF_2}{FS_1}+
    \frac{FS_1-FS_2}{FS_1}
    \underline{\Delta}
    \le\Delta^1_2(C_1)\le
    \frac{RF_2}{FS_1}+
    \frac{FS_1-FS_2}{FS_1}
    \overline{\Delta}.
    \end{equation*}
\end{remark}

\begin{remark}
\label{remark_signrestriction}
    The bounds in Equations \eqref{specificlowerbound2} and \eqref{specificupperbound2} simplify under sign restrictions for the treatment effects $\Delta^0_2(g)$. For example, if we assume causal effects are nonnegative ($\underline{\Delta}=0$), then $RF_2/FS_1$ would be the lower bound or upper bound (depending on whether $FS_2$ is lower than $FS_1$). In particular, if $FS_2\le FS_1$, $RF_2/FS_1$ is the lower bound.
\end{remark}

The bounds in Equations \eqref{generallowerbound2} and \eqref{generalupperbound2} are valid without any assumption other than irreversible treatment (Assumption \ref{irreverse}) and the basic conditions for IV validity (Assumption \ref{basicvalidity}). When treatment effects for the groups that contaminate $RF_2$ are homogeneous given period and treatment length, the tighter bounds in Equations \eqref{specificlowerbound2} and \eqref{specificupperbound2} are valid. For the bounds in Equations \eqref{generallowerbound2} and \eqref{generalupperbound2}, the smaller the probability of late switching into treatment, the tighter the bounds. For the bounds in Equations \eqref{specificlowerbound2} and \eqref{specificupperbound2}, the smaller the change in the first stage, the tighter the bounds. Appendix \ref{partialid_generalboundste} provides bounds without assuming a nonpositive lower bound and a nonengative upper bound for treatment effects.

\section{$T$-periods setting} \label{Sec:T periods}

The results from Section \ref{Sec:2 periods} generalize for settings with an arbitrary number of periods. Consider a setting with $T$ periods of time and let $\T\coloneqq\{1,...,T\}$. The definitions of $RF_t$, $FS_t$, and latent groups extend naturally for this setting with $T$ periods. Assumption \ref{irreverse} becomes: 

\begin{assumption}[Irreversible Treatment]
\label{irreverse_T}
For all $t\in\T\setminus\{T\}$, $D_{i,t}=1\implies D_{i,t+1}=1$.
\end{assumption}

Given irreversible treatment, denote potential outcomes by $Y_{i,t}(0,z)$, and $Y_{i,t}(1,\tau,z)$ depending on whether the observation has never been treated, or on whether it has been first treated at period $t-\tau$. We consider an extension of Assumption \ref{basicvalidity} for settings with $T$ periods. Once more, note that it only requires relevance and monotonicity in the first period. 

\begin{assumption}
\label{basicvalidity_T}
The following hold:
\begin{enumerate}
    \item \textit{Exclusion}: For each $t\in\T$ and $z\in\{0,1\}$, $Y_{i,t}(0, z)=Y_{i,t}(0)$ and $Y_{i,t}(1,\tau, z)=Y_{i,t}(1,\tau)$ for all $\tau\in\{0,...,t-1\}$.
    \item \textit{Independence}: $\big(Y_{i,t}(0), Y_{i,t}(1,0),...,Y_{i,t}(1,t-1), D_{i,1}(1), D_{i,1}(0),...,D_{i,t}(1), D_{i,t}(0)\big)$ is independent of $Z_i$ for all $t\in\T$.
    \item \textit{Relevance at $t=1$}: $FS_1\neq0$.
    \item \textit{Monotonicity at $t=1$}: $\mathbb{P}(F_1)=0$.
\end{enumerate}
\end{assumption}

In this case, we are interested in estimating the treatment effects $\Delta_t^{t-1}(C_1)$, which represent the local average treatment effects at time $t$ of being treated $t-1$ periods before (that is, when treatment started at $t=1$), for the first-period compliers. As before, the per-period IV estimand identifies  $\Delta_t^{t-1}(C_1)$ under Assumption \ref{basicvalidity_T} if there is static compliance. However, this would not be the case when compliance is dynamic.

\subsection{Decomposition of RF and IV estimands with $T$ periods}

To generalize Proposition \ref{decomposition2} for settings with $T$ periods, write $C_{t:t'}$ for observations that are compliers from $t$ to $t'$, with analogous notation for defiers and never-takers. We only keep track of the first period in which observations are always-takers because always-takers in a given period are always-takers in all following periods. Moreover, define the following sets:
\begin{equation*}
\begin{split}
&\G^+_2\coloneqq\big\{(NT_1, C_2)\big\},\\
&\G^-_2\coloneqq\big\{(C_1, AT_2), (NT_1, F_2)\big\},
\end{split}
\end{equation*}
and, for each $t\in\T\setminus\{1,2\}$,
\begin{equation*}
\begin{split}
    &\G^+_t\coloneqq
    \big\{ 
    (NT_{1:t-1}, C_t), 
    (NT_{1:\ell-1}, F_{\ell:t-1}, AT_{t}): \ell=2,...,t-1
    \big\},\\
    &\G^-_t\coloneqq
    \big\{
    (C_{1:t-1}, AT_t), 
    (NT_{1:t-1}, F_t),
    (NT_{1:\ell-1}, C_{\ell:t-1}, AT_{t}): \ell=2,...,t-1
    \big\}.
\end{split}
\end{equation*}

Assumption \ref{irreverse_T} implies that, for each $t\in\T\setminus\{1\}$, the latent groups in $\G^+_t$ are the ones that switch into treatment at $t$ when $Z_i=1$ and the latent groups in $\G^-_t$ are the ones that switch into treatment at $t$ when $Z_i=0$. The following proposition generalizes the decomposition of per-period reduced forms and first stages.

\begin{proposition}
\label{rffsestimands}
Under Assumptions \ref{irreverse_T} and \ref{basicvalidity_T}, for each $t\in\T\setminus\{1\}$,
\begin{equation}
\label{lambdat}
    RF_t
    =
    \Pb\left(C_1\right)
    \Delta_t^{t-1}(C_1)
    -
    \sum_{k=2}^{t}\sum_{g\in\G^-_k}
    \Pb\left(g\right)
    \Delta_t^{t-k}\left(g\right)
    +
    \sum_{k=2}^{t}\sum_{g\in\G^+_k}
    \Pb\left(g\right)
    \Delta_t^{t-k}\left(g\right)
\end{equation}
and 
\begin{equation}
\label{gammat}
FS_t=
    \Pb\left(C_{1}\right)
    -
    \sum_{k=2}^{t}\sum_{g\in\G^-_k}
    \Pb\left(g\right)
    +
    \sum_{k=2}^{t}\sum_{g\in\G^+_k}
    \Pb\left(g\right).
\end{equation}
\end{proposition}
\begin{proof}
    See Appendix \ref{proofrffsestimands}.
\end{proof} 

\begin{corollary}
\label{IVestimand}
Under Assumptions \ref{irreverse_T} and \ref{basicvalidity_T}, for any $t\in\T\setminus\{1\}$ such that $FS_t\neq0$, $RF_t/FS_t$ is a linear combination of the causal effects in Equation \eqref{lambdat} in which the weights sum to one but some of them may be negative. A sufficient condition for the existence of negative weights at $t$ is the existence of $k\in\{2,...,t\}$ such that $FS_k<FS_{k-1}$. Moreover, the causal effects that are negatively weighted in $RF_t/FS_t$ are the same as in $RF_t$ if, and only if, $FS_t>0$. 
\end{corollary}
\begin{proof}
    See Appendix \ref{proofcorollaryIVestimand}.
\end{proof}

\subsection{Point identification with $T$ periods}
\label{sec:T periods point id}

For each $t\in\T\setminus\{1\}$, define
\begin{equation*}
    \G_t\coloneqq\G^+_t\cup\G^-_t,
\end{equation*}
the set of latent groups that switch into treatment at $t$ and contaminate the reduced form. The following assumption generalizes Assumption \ref{Assumption_id}.

\begin{assumption}
\label{Assumption_id_T}
For all $t\in\T$ and $\tau\in\{0,...,t-1\}$, $\Delta^\tau_t(C_1)=\Delta^\tau(C_1)$. Moreover, for each $t\in\T\setminus\{1\}$ and $\tau\in\{0,...,t-2\}$, for any latent group $g\in\G_{t-\tau}$ such that $\Pb(g)>0$, $\Delta^\tau(C_1)=\Delta^\tau_t(g)$. 
\end{assumption}

Proposition \ref{identification} below formalizes the identification result. To state it, consider matrix notation. Let $\mathbf{RF}\coloneqq\left(RF_1,...,RF_T\right)'$. For each $t\in\T\setminus\{1\}$, define $\rho_t\coloneqq\Pb\left(D_{i,t}>D_{i,t-1}\middle|Z_i=0\right)\\ -\Pb\left(D_{i,t}>D_{i,t-1}\middle|Z_i=1\right)$, the difference between the probability of switching into treatment for $Z_i=0$ and $Z_i=1$ observations, which equals $FS_{t-1}-FS_t$ due to the irreversibility of treatment (Assumption \ref{irreverse_T}). Moreover, let
\begin{equation*}
    \mathbf{P}\coloneqq
    \begin{bmatrix}
     FS_1 & 0 & \hdots & 0\\
    -\rho_2 & FS_1 & \hdots & 0\\
    \vdots & \vdots & \ddots & \vdots \\
    -\rho_{T} & -\rho_{T-1} & \hdots & FS_1
    \end{bmatrix},
\end{equation*}
which is a lower triangular $T\times T$ matrix. Note that $\mathbf{P}$ is invertible provided that the instrument is relevant in the first period.

\begin{proposition}
\label{identification}
Suppose Assumptions \ref{irreverse_T} and \ref{basicvalidity_T} hold. Under Assumption \ref{Assumption_id_T},
\begin{equation}
\label{solvedelta}
\mathbf{\Delta}=\mathbf{P}^{-1}\mathbf{RF},
\end{equation}
where $\mathbf{\Delta}\coloneqq\left(\Delta^0(C_1),...,\Delta^{T-1}(C_1)\right)'$.
\end{proposition}
\begin{proof}
    See Appendix \ref{proofidentification}.
\end{proof}

\subsection{Partial identification with $T$ periods}

 In the general $T$-periods setting, dynamic LATEs are partially identified in every period for which the treatment effects are bounded (which, again, nests settings with bounded outcomes). Proposition \ref{tebounds_T} generalizes Proposition \ref{tebounds2}. Appendix \ref{partialid_generalboundste} provides general bounds without requiring the lower bound (upper bound) for the treatment effects to be nonpositive (nonnegative).

\begin{proposition}
\label{tebounds_T}
Suppose Assumptions \ref{irreverse_T} and \ref{basicvalidity_T} hold. If, for $t\in\T\setminus\{1\}$, there exist $\underline{\Delta}_t, \overline{\Delta}_t\in\mathbb{R}$, with $\underline{\Delta}_t\le0\le\overline{\Delta}_t$, such that, for each $\tau\in\{0,...,t-2\}$, if $g\in\G_{t-\tau}$ and $\Pb(g)>0$, $\underline{\Delta}_t\le\Delta^\tau_t(g)\le\overline{\Delta}_t$, then a lower bound for $\Delta^{t-1}_t(C_1)$ is given by
\begin{equation}
\label{generallowerbound_T}
\begin{split}
    \frac{RF_t}{FS_1}
    +
    \Pb\left(D_{i,t}>D_{i,1}\middle|Z_i=0\right)
    \frac{\underline{\Delta}_t}{FS_1}
    -
    \Pb\left(D_{i,t}>D_{i,1}\middle|Z_i=1\right)
    \frac{\overline{\Delta}_t}{FS_1}
\end{split}
\end{equation}
and an upper bound is given by
\begin{equation}
\label{generalupperbound_T}
\begin{split}
    \frac{RF_t}{FS_1}
    +
    \Pb\left(D_{i,t}>D_{i,1}\middle|Z_i=0\right)
    \frac{\overline{\Delta}_t}{FS_1}
    -
    \Pb\left(D_{i,t}>D_{i,1}\middle|Z_i=1\right)
    \frac{\underline{\Delta}_t}{FS_1}.
\end{split}
\end{equation}
If, in addition to the conditions above, for each $\tau\in\{0,...,t-2\}$, for all $g,g'\in\G_{t-\tau}$ with $\Pb(g)>0$ and $\Pb(g')>0$, $\Delta^\tau_t(g)=\Delta_t^\tau(g')$, then
\begin{equation}
\label{specificlowerbound_T}
    \frac{RF_t}{FS_1} 
    +
    \underline{\Delta}_t
    \frac{\left(FS_1-FS_t\right)}{FS_1}
    +
    \left(\overline{\Delta}_t-\underline{\Delta}_t\right)
    \sum_{k=2}^t
    \mathbf{1}\left(FS_{k-1}< FS_k\right)
    \frac{FS_{k-1}-FS_k}{FS_1}
\end{equation}
is a lower bound for $\Delta_t^{t-1}(C_1)$ and
\begin{equation}
\label{specificupperbound_T}
    \frac{RF_t}{FS_1} 
    +
    \overline{\Delta}_t
    \frac{\left(FS_1-FS_t\right)}{FS_1}
    +
    \left(\underline{\Delta}_t-\overline{\Delta}_t\right)
    \sum_{k=2}^t
    \mathbf{1}\left(FS_{k-1}< FS_k\right)
    \frac{FS_{k-1}-FS_k}{FS_1}
\end{equation}
is an upper bound for $\Delta_t^{t-1}(C_1)$. These bounds are (weakly) tighter than the previous ones. 
\end{proposition}
\begin{proof}
    See Appendix \ref{partialIDproof}.
\end{proof}

\begin{remark}
\label{remarkpartialid}
    The points in Remarks \ref{remark_nonexistentgroups} and \ref{remark_signrestriction} generalize. Assuming that $\Pb(NT_{1:k-1}, C_k)=\Pb(NT_{1:k-1}, F_k)=0$ for all $k\in\{2,...,t\}$ implies that the conditions in Proposition \ref{tebounds_T} for tighter bounds hold at $t$ and that first stages are nonincreasing (up to $t$). Under a sign restriction for treatment effects, if first stages are monotonic and the condition for tighter bounds holds, then $RF_t/FS_1$ is one of the bounds (whether it is the lower or upper bound depends on first stages being decreasing or increasing).
\end{remark}

\section{Conclusion}
\label{sec:conclusion}
We consider the identification of dynamic causal effects of an irreversible binary treatment when the only source of exogenous variation is a time-invariant binary instrument. Under a dynamic extension of standard IV assumptions, we decompose the per-period IV estimands as a weighted sum of causal effects for different latent groups and treatment exposures. Even though the weights given for causal effects sum to one, some may be negative, which greatly restricts even a weakly causal interpretation (in \textcite{blandholetal2022}'s sense) of per-period IV estimands. In particular,  per-period IV estimands may be negative even when all treatment effects are positive. A sufficient condition for the existence of negative weights is that the first stage decreases with time.

Dynamic LATEs are shown to be identified by the per-period IV estimands under strong assumptions, including causal effects not depending on the time since treatment. We consider an alternative set of assumptions allowing unrestricted heterogeneity in the time-since-treatment dimension but requiring homogeneity in the calendar-time dimension. Under this alternative assumption, dynamic LATEs are identified recursively by correcting each period's bias using previously identified effects. In an extension of \textcite{ischemia}, this identifies exposure effects allowing for defiance after the first period. This flexibility is useful in settings where, for example, those lottery assigned to treatment that did not get treated in the first period face restrictions in later periods.

For settings in which both homogeneity assumptions may be too restrictive, we show how dynamic LATEs can be partially identified without any homogeneity conditions on the causal effect. We also show how to tighten these bounds by imposing cross-group homogeneity assumptions while allowing for unrestricted heterogeneity across both calendar time and exposure dimensions.


\clearpage
\printbibliography

\clearpage
\appendix
\section{Proofs}
\label{appendixproofs}

\subsection{Proof of Proposition \ref{rffsestimands}}
\label{proofrffsestimands}

Fix $t\in\T\setminus\{1\}$. Under Assumption \ref{basicvalidity_T}, the only latent groups that do not have equal potential outcomes (in expectation) when assigned to different instrument values are the ones that would behave differently if assigned to $Z_i=1$ or $Z_i=0$. Thus, observations that are always-takers in all periods, observations that are never-takers in all periods up to $t$, and obsevations such that $(NT_{1:k-1}, AT_{k})$ for some $k\in\{2,...,t\}$ do not show up in our decomposition. The terms related to them cancel out.

Assumptions \ref{irreverse_T} and \ref{basicvalidity_T} imply that $C_1$, $(NT_{1:k-1}, C_{k})$ or $(NT_{1:k-1}, F_{k})$ with $k\in\{2,...,t\}$ are the only groups that can have different potential treatment status depending on $Z_i$ at $t$. Moreover, at each $k\in\{2,...,t\}$, $NT_{k-1}$ observations' behavior parallels the behavior of all observations in the first period, except that we allow for defiance. In particular, because of Assumption \ref{irreverse_T}, treatment access for $(NT_{1:k-1}, C_{k})$ and $(NT_{1:k-1}, F_{k})$ groups, with $k\in\{2,...,t\}$, has a dynamic that is analogous to the one for the $C_1$ group. Therefore, it suffices to consider the decomposition of $\E[Y_{i,t}|Z_i=1, C_1]-\E[Y_{i,t}|Z_i=0,C_1]$. Decomposition of the other terms follows from similar calculations, noting that defiers enter $RF_t$ with opposite signs. 

From Assumption \ref{irreverse_T}, $C_1$ observations with $Z_i=1$ are treated in all periods and so 
\begin{equation}
\label{decompositionC1Z1}
\E[Y_{i,t}|Z_i=1, C_1]=\E[Y_{i,t}(1,t-1)|C_1] 
\end{equation}
follows from Assumption \ref{basicvalidity_T}. To relate $\E[Y_{i,t}|Z_i=0,C_1]$ to potential outcomes, we need to consider all possible latent group histories $C_1$ observations can take up to $t$. Under Assumption \ref{irreverse_T}, these histories have the form $(C_{1:k-1}, AT_{k})$ with $k\in\{2,...,t\}$ or $C_{1:t}$. Working forwardly and applying Assumption \ref{basicvalidity_T}, we get:
\begin{equation}
\label{decompositionC1Z0}
    \begin{split}
    \E\left[Y_{i,t}\middle|Z_i=0, C_1\right]&
    =\Pb\left(AT_{2}\middle|C_1\right)
    \E\left[Y_{i,t}(1, t-2)\middle|C_1, AT_2\right]\\&
    +\Pb\left(C_{2}\middle|C_1\right)
    \E\left[Y_{i,t}\middle|Z_i=0, C_{1:2}\right]\\
    &=\Pb\left(AT_{2}\middle|C_1\right)
    \E\left[Y_{i,t}(1, t-2)\middle|C_1, AT_2\right]\\
    &\;
    \begin{split}
    +\;\Pb\left(C_{2}\middle|C_1\right)
    \Big\{&
    \Pb\left(AT_{3}\middle|C_{1:2}\right)
    \E\left[Y_{i,t}(1, t-3)\middle|C_{1:2}, AT_3\right]\\
    &
    +\Pb\left(C_{3}\middle|C_{1:2}\right)
    \Big[
    \Pb\left(AT_{4}\middle|C_{1:3}\right)
    \E\left[Y_{i,t}(1, t-4)\middle|C_{1:3}, AT_4\right]\\
    &\;\;\;\;\;
    \begin{split}
    +...\Pb\left(C_{t-1}\middle|C_{1:t-2}\right)
    \Big(
    &\Pb\left(AT_{t}\middle|C_{1:t-1}\right)
    \E\left[Y_{i,t}(1,0)\middle|C_{1:t-1}, AT_t\right]\\
    &+\Pb\left(C_{t}\middle|C_{1:t-1}\right)
    \E\left[Y_{i,t}(0)\middle|C_{1:t}\right]
    \Big)...\Big]\Big\}.
    \end{split}
    \end{split}
\end{split}
\end{equation}

Noting that $\E\left[Y_{i,t}(1,0)\middle|C_{1:t-1}, AT_t\right]=\E\left[Y_{i,t}(0)\middle|C_{1:t-1}, AT_t\right] + \Delta_t^0(C_{1:t-1}, AT_t)$, it follows from the Law of Iterated Expectations that the last term in parenthesis in the expression for $\E[Y_{i,t}|Z_i=0, C_1]$ equals $\Pb\left(AT_{t}\middle|C_{1:t-1}\right)
\Delta_t^0(C_{1:t-1}, AT_t)
+\E\left[Y_{i,t}(0)\middle|C_{1:t-1}\right]$.

Repeating this process backwards, we obtain:
\begin{equation*}
    \E\left[Y_{i,t}\middle|Z_i=0, C_1\right]=
    \E\left[Y_{i,t}(0)\middle|C_1\right]
    +\sum_{k=2}^{t}
    \left(
    \prod_{\ell=2}^{k-1}
    \Pb\left(C_{\ell}\middle|C_{1:\ell-1}\right)
    \right)
    \Pb\left(AT_{k}\middle|C_{1:k-1}\right)
    \Delta_t^{t-k}\left(C_{1:k-1}, AT_k\right),
\end{equation*}
under the convention that $\prod_{\ell=2}^{1}
...=1$. Lastly, write the product of probabilities as a joint probability to get:
\begin{equation*}
    \E\left[Y_{i,t}\middle|Z_i=0, C_1\right]
    =
    \E\left[Y_{i,t}(0)\middle| C_1\right]+
    \sum_{k=2}^{t}
    \Pb\left(
    C_{1:k-1}, AT_k\middle|C_1\right)
    \Delta_t^{t-k}\left(C_{1:k-1}, AT_k\right),
\end{equation*}
which implies:
\begin{equation*}
    \E\left[Y_{i,t}\middle|Z_i=1, C_1\right]-\E\left[Y_{i,t}\middle|Z_i=0, C_1\right]
    =
    \Delta_t^{t-1}(C_1)
    -
    \sum_{k=2}^{t}
    \Pb\left(
    C_{1:k-1}, AT_k\middle|C_1\right)
    \Delta_t^{t-k}\left(C_{1:k-1}, AT_k\right).
\end{equation*}

Computing the analogous decomposition for each of the other histories and accounting for the probability of each of them, we get:
\begin{equation}
\label{reducedform_extensive}
\begin{split}
    RF_t
    &=
    \Pb\left(C_1\right)
    \Delta_t^{t-1}(C_1)\\
    &-
    \sum_{k=2}^{t}
    \Pb\left(
    C_{1:k-1}, AT_k\right)
    \Delta_t^{t-k}\left(C_{1:k-1}, AT_k\right)
    \\
    &
    \begin{split}
    +
    \sum_{k=2}^{t}
    \Bigg[\Pb\left(NT_{1:k-1},C_k\right)
    &\Delta_t^{t-k}\left(NT_{1:k-1},C_k\right)\\
    &-\sum_{\ell=k+1}^t
    \Pb\left(NT_{1:k-1},C_{k:\ell-1}, AT_\ell\right)
    \Delta_t^{t-\ell}\left(NT_{1:k-1}, C_{k:\ell-1}, AT_\ell\right)
    \Bigg]
    \end{split}
    \\
    &
    \begin{split}
    -
    \sum_{k=2}^{t}
    \Bigg[\Pb\left(NT_{1:k-1},F_k\right)
    &\Delta_t^{t-k}\left(NT_{1:k-1},F_k\right)\\
    &-\sum_{\ell=k+1}^t
    \Pb\left(NT_{1:k-1},F_{k:\ell-1}, AT_\ell\right)
    \Delta_t^{t-\ell}\left(NT_{1:k-1}, F_{k:\ell-1}, AT_\ell\right)
    \Bigg],
    \end{split}
\end{split}
\end{equation}
under the convention that $\sum_{\ell=t+1}^t...=0$. Note that $\sum_{k=2}^t\sum_{\ell=k+1}^t...$ under the convention $\sum_{\ell=t+1}^t...=0$ can be written as $\sum_{\ell=2}^t\sum_{k=2}^{\ell-1}...$ under the convention $\sum_{k=2}^1...=0$. Thus, rearranging Equation \eqref{reducedform_extensive} and changing the index in the double sums (so that the outer summation  is indexed by $k$ and the inner one by $\ell$ with appropriate adjustment in the subscripts), we obtain:
\begin{equation}
\begin{split}
\label{reducedform_rear}
    RF_t
    &=
    \Pb\left(C_1\right)
    \Delta_t^{t-1}(C_1)
    -
    \sum_{k=2}^{t}
    \Pb\left(
    C_{1:k-1}, AT_k\right)
    \Delta_t^{t-k}\left(C_{1:k-1}, AT_k\right)
    \\
    &
    +
    \sum_{k=2}^{t}
    \Bigg[\Pb\left(NT_{1:k-1},C_k\right)
    \Delta_t^{t-k}\left(NT_{1:k-1},C_k\right)
    -\Pb\left(NT_{1:k-1},F_k\right)
    \Delta_t^{t-k}\left(NT_{1:k-1},F_k\right)\Bigg]
    \\
    &
    \begin{split}
    -\sum_{k=2}^t\sum_{\ell=2}^{k-1}\Bigg[
    \Pb\left(NT_{1:\ell-1},C_{\ell:k-1}, AT_k\right)
    &\Delta_t^{t-k}\left(NT_{1:\ell-1}, C_{\ell:k-1}, AT_k\right)\\
    -&\Pb\left(NT_{1:\ell-1},F_{\ell:k-1}, AT_k\right)
    \Delta_t^{t-k}\left(NT_{1:\ell-1}, F_{\ell:k-1}, AT_k\right)
    \Bigg].
    \end{split}
\end{split}
\end{equation}

The result as stated in Equation \eqref{lambdat} follows from noting that for each $k\in\{2,...,t\}$, any group $g$ for which the causal effect $\Delta^{t-k}_{t}(g)$ appears in Equation \eqref{reducedform_rear} multiplied by a negative (respectively, positive) probability is such that $g\in\G^-_k$ (respectively, $g\in\G^+_k$). 

For $FS_t$, we get from an analogous argument:
\begin{equation}
\begin{split}
    FS_t
    &=
    \Pb\left(C_1\right)
    -
    \sum_{k=2}^{t}
    \Pb\left(
    C_{1:k-1}, AT_k\right)
    +
    \sum_{k=2}^{t}
    \Bigg[\Pb\left(NT_{1:k-1},C_k\right)
    -\Pb\left(NT_{1:k-1},F_k\right)
    \Bigg]
    \\
    &
    -\sum_{k=2}^t\sum_{\ell=2}^{k-1}\Bigg[
    \Pb\left(NT_{1:\ell-1},C_{\ell:k-1}, AT_k\right)
    -\Pb\left(NT_{1:\ell-1},F_{\ell:k-1}, AT_k\right)
    \Bigg],
\end{split}
\end{equation}
under the convention that $\sum_{\ell=2}^1...=0$. Again, the result as stated in Equation \eqref{gammat} follows from the definition of the sets $\G^-_k$'s and $\G^+_k$'s.

\subsection{Proof of Corollary \ref{IVestimand}}
\label{proofcorollaryIVestimand}
That $RF_t/FS_t$ is a linear combination of the causal effects in $RF_t$ is straightforward. That the weights in the IV estimand sum to one follows from noting that the sum of the probabilities in $RF_t$ equals $FS_t$. For any given $k\in\{2,...,t\}$, we have that 
\begin{equation*}
    FS_{k}-FS_{k-1}=
    -\sum_{g\in\G^-_k}\Pb(g)
    +\sum_{g\in\G^+_k}\Pb(g)<0
\end{equation*}
only when there exists $g\in\G^-_k$ such that $\Pb(g)>0$, which implies that there is at least one causal effect that enters $RF_t$ multiplied by a negative probability, which in turn implies a negative weight in the IV estimand at $t$. Lastly, $FS_t>0$ is a necessary and sufficient condition for the negatively weighted causal effects in $RF_t$ and $RF_t/FS_t$ to be the same because the sign of the weights in the IV estimand equals the sign of the weights in $RF_t$ times the sign of $FS_t$. 

\subsection{Proof of Proposition \ref{identification}}
\label{proofidentification}

For any $t\in\T\setminus\{1\}$,
\begin{equation*}
    \rho_t=FS_{t-1}-FS_t=\sum_{g\in\G^-_t}\Pb(g)
    -\sum_{g\in\G^+_t}\Pb(g).
\end{equation*} 

Under Assumption \ref{Assumption_id_T}, for any given $t\in\T\setminus\{1\}$, $RF_t$ becomes
\begin{equation*}
\begin{split}
\label{RFforpartialid}
RF_t
    &=
    \Pb\left(C_1\right)
    \Delta^{t-1}(C_1)
    -
    \left[\sum_{k=2}^{t}
    \left(\sum_{g\in\G^-_k}
    \Pb\left(g\right)
    -\sum_{g\in\G^+_k}
    \Pb\left(g\right)
    \right)
    \Delta^{t-k}\left(C_1\right)\right]\\
    &=
    \Pb\left(C_1\right)
    \Delta^{t-1}(C_1)
    -
    \sum_{k=2}^{t}
    \rho_k
    \Delta^{t-k}\left(C_1\right),
\end{split}
\end{equation*}
which implies the linear system $\mathbf{RF}=\mathbf{P}\mathbf{\Delta}$ if we recall that $FS_1=\Pb(C_1)$ and that $RF_1=\Pb(C_1)\Delta^{0}(C_1)$ under Assumption \ref{Assumption_id_T}. The desired result follows from $\mathbf{P}$ being invertible under Assumption \ref{basicvalidity_T}.

\subsection{Proof of Proposition \ref{tebounds_T}}
\label{partialIDproof}

For the bounds that are valid only assuming \ref{irreverse_T} and \ref{basicvalidity_T}, we prove the more general version (as stated in Appendix \ref{partialid_generalboundste}). Fix $t\in\T\setminus\{1\}$. Rearranging the reduced form (Equation \eqref{lambdat}):
\begin{equation}
\label{RFrearbounds}
\begin{split}
    \Pb\left(C_1\right)
    \Delta_t^{t-1}(C_1)
    &=
    RF_t
    +
    \sum_{k=2}^{t}\sum_{g\in\G^-_k}
    \Pb\left(g\right)
    \Delta_t^{t-k}\left(g\right)
    -
    \sum_{k=2}^{t}\sum_{g\in\G^+_k}
    \Pb\left(g\right)
    \Delta_t^{t-k}\left(g\right)\\
    &\ge
    RF_t
    +
    \underline{\Delta}_t
    \sum_{k=2}^{t}\sum_{g\in\G^-_k}
    \Pb\left(g\right)
    -
    \overline{\Delta}_t
    \sum_{k=2}^{t}\sum_{g\in\G^+_k}
    \Pb\left(g\right)
\end{split}
\end{equation}

Notice that $\sum_{k=2}^{t}\sum_{g\in\G^-_k} \Pb\left(g\right)\le\Pb\left(D_{i,t}>D_{i,1}\middle|Z_i=0\right)$ because there are latent groups that switch into treatment after the first period when $Z_i=0$ that are not included in the sets $\G^-_k$ for any $k\in\{2,...,t\}$ (namely, the $(NT_{1:k-1}, AT_k)$ with $k\in\{2,...,t\}$). Moreover, $\sum_{k=2}^{t}\sum_{g\in\G^-_k} \Pb\left(g\right)\ge \max\left\{FS_1-FS_t,0\right\}$. Also, $\max\left\{FS_t-FS_1,0\right\}\le\sum_{k=2}^{t}\sum_{g\in\G^+_k} \Pb\left(g\right) \le\Pb\left(D_{i,t}>D_{i,1}\middle|Z_i=1\right)$. Thus, we can get a lower bound for the expression in the second row of Equation \eqref{RFrearbounds} by bounding the sum of probabilities, which implies the the following lower bound for $\Pb\left(C_1\right) \Delta_t^{t-1}(C_1)$:
\begin{equation*}
\begin{split}
    RF_t
    &+\mathbf{1}\left(\underline{\Delta}_t<0\right)
    \Pb\left(D_{i,t}>D_{i,1}\middle|Z_i=0\right)
    \underline{\Delta}_t
    +\mathbf{1}\left(\underline{\Delta}_t\ge0\right)
    \max\left\{FS_1-FS_t,0\right\}
    \underline{\Delta}_t
    \\
    &-\mathbf{1}\left(\overline{\Delta}_t\ge0\right)
    \Pb\left(D_{i,t}>D_{i,1}\middle|Z_i=1\right)
    \overline{\Delta}_t
    -\mathbf{1}\left(\overline{\Delta}_t<0\right)
    \max\left\{FS_t-FS_1,0\right\}
    \overline{\Delta}_t
    ,
\end{split}
\end{equation*}
from which the lower bound in Equation \eqref{generallowerbound_T} follows directly since $\Pb(C_1)=FS_1>0$ under Assumption \ref{basicvalidity_T}. The argument for the upper bound is analogous. 

To obtain the bounds under the condition that for each $\tau\in\{0,...,t-2\}$, for all $g,g'\in\G_{t-\tau}$ with $\Pb(g)>0$ and $\Pb(g')>0$, $\Delta^\tau_t(g)=\Delta_t^\tau(g')$, note that under such condition $RF_t$ (Equation \eqref{lambdat}) becomes
\begin{equation*}
    \Pb\left(C_1\right)
    \Delta_t^{t-1}(C_1)
    =RF_t
    +
    \sum_{k=2}^{t}
    \rho_k
    \Delta^{t-k}_t(*),
\end{equation*}
where, for a given $k\in\{2,...,t\}$, $\Delta^{t-k}_t(*)\in\left[\underline{\Delta}_t, \overline{\Delta}_t\right]$ equals $\Delta^{t-k}_t(g)$ for all $g\in\G_k$. Then, because for any $k\in\{2,...,t\}$, $\rho_k=FS_{k-1}-FS_k$,
\begin{equation}
\label{RFsearbounds}
\begin{split}
    \Pb\left(C_1\right)
    \Delta_t^{t-1}(C_1)
    &\ge 
    RF_t
    +
    \sum_{k=2}^t
    \mathbf{1}\left(FS_{k-1}\ge FS_k\right)
    \rho_k
    \underline{\Delta}_t
    +
    \sum_{k=2}^t
    \mathbf{1}\left(FS_{k-1}< FS_k\right)
    \rho_k
    \overline{\Delta}_t\\
    &=
    RF_t
    +
    \underline{\Delta}_t
    \sum_{k=2}^t
    \rho_k
    +
    \left(\overline{\Delta}_t-\underline{\Delta}_t\right)
    \sum_{k=2}^t
    \mathbf{1}\left(FS_{k-1}< FS_k\right)
    \rho_k
\end{split}
\end{equation}
and the upper bound follows from an analogous argument. The bounds as stated in the proposition follow from $\sum_{k=2}^t\rho_k=FS_1-FS_t$. To prove that these later bounds are tighter, from comparing Equations \eqref{RFrearbounds} and \eqref{RFsearbounds}, we note that a sufficient condition for the lower bound to be tighter is 
\begin{equation*}
\begin{split}
    &\left[\sum_{k=2}^t
    \mathbf{1}\left(FS_{k-1}\ge FS_k\right)
    \rho_k-\sum_{k=2}^{t}\sum_{g\in\G^-_k}
    \Pb\left(g\right)
    \right]
    \underline{\Delta}_t
    \\+&
    \left[\sum_{k=2}^t
    \mathbf{1}\left(FS_{k-1}< FS_k\right)
    \rho_k+\sum_{k=2}^{t}\sum_{g\in\G^+_k}
    \Pb\left(g\right)
    \right]
    \overline{\Delta}_t
    \ge0,
\end{split}
\end{equation*}
which is equivalent to
\begin{equation*}
\begin{split}
    &\left[\sum_{k=2}^t\rho_k
    -\sum_{k=2}^{t}\sum_{g\in\G^-_k}
    \Pb\left(g\right)
    +\sum_{k=2}^{t}\sum_{g\in\G^+_k}
    \Pb\left(g\right)
    \right]
    \underline{\Delta}_t
    \\+&
    \left[\sum_{k=2}^t
    \mathbf{1}\left(FS_{k-1}< FS_k\right)
    \rho_k+\sum_{k=2}^{t}\sum_{g\in\G^+_k}
    \Pb\left(g\right)
    \right]
    \left(\overline{\Delta}_t
    -\underline{\Delta}_t\right)
    \ge0\\
    \iff&
    \left[\sum_{k=2}^t
    \mathbf{1}\left(FS_{k-1}< FS_k\right)
    \rho_k+\sum_{k=2}^{t}\sum_{g\in\G^+_k}
    \Pb\left(g\right)
    \right]
    \left(\overline{\Delta}_t
    -\underline{\Delta}_t\right)
    \ge0
\end{split}
\end{equation*}
since $\sum_{k=2}^t\rho_k-\sum_{k=2}^{t}\sum_{g\in\G^-_k}\Pb\left(g\right)+\sum_{k=2}^{t}\sum_{g\in\G^+_k}\Pb\left(g\right)=0$. Because $\overline{\Delta}_t-\underline{\Delta}_t\ge0$ and 
\begin{equation}
\label{strictconditionlower}
\begin{split}
    -\sum_{k=2}^t\mathbf{1}\left(FS_{k-1}< FS_k\right)\rho_k
    &=
    \sum_{k=2}^t\mathbf{1}\left(FS_{k-1}< FS_k\right)\left[
    \sum_{g\in\G^+_k}\Pb(g)
    -\sum_{g\in\G^-_k}\Pb(g)\right]\\
    &\le 
    \sum_{k=2}^t\mathbf{1}\left(FS_{k-1}< FS_k\right)
    \sum_{g\in\G^+_k}\Pb(g)\\
    &\le 
    \sum_{k=2}^t\sum_{g\in\G^+_k}\Pb(g),
\end{split}
\end{equation}
the condition is verified. Once more, the argument for the upper bound is analogous.

\newpage
\section{Partial identification with general bounds on treatment effects}
\label{partialid_generalboundste}

Proposition \ref{tebounds_T_general} states a version of Proposition \ref{tebounds_T} in which lower bounds for the treatment effects can be positive and upper bounds can be negative. It only extends the more general bounds presented in Proposition \ref{tebounds_T} (Equations \eqref{generallowerbound_T} and \eqref{generalupperbound_T}) because the bounds in Equations \eqref{specificlowerbound_T} and \eqref{specificupperbound_T} are generally valid (and continue to be weakly tighter). Appendix \ref{partialIDproof} gives a proof of this proposition.

\begin{propositionp}{\ref*{tebounds_T}$'$}
\label{tebounds_T_general}
Suppose Assumptions \ref{irreverse_T} and \ref{basicvalidity_T} hold. If, for $t\in\T\setminus\{1\}$, there exist $\underline{\Delta}_t, \overline{\Delta}_t\in\mathbb{R}$ such that, for all $\tau\in\{0,...,t-2\}$, if $g\in\G_{t-\tau}$ and $\Pb(g)>0$, $\underline{\Delta}_t\le\Delta^\tau_t(g)\le\overline{\Delta}_t$, then a lower bound for $\Delta^{t-1}_t(C_1)$ is given by
\begin{equation*}
\begin{split}
    \frac{RF_t}{FS_1}
    &+\mathbf{1}\left(\underline{\Delta}_t<0\right)
    \Pb\left(D_{i,t}>D_{i,1}\middle|Z_i=0\right)
    \frac{\underline{\Delta}_t}{FS_1}
    +\mathbf{1}\left(\underline{\Delta}_t\ge0\right)
    \max\left\{FS_1-FS_t,0\right\}
    \frac{\underline{\Delta}_t}{FS_1}
    \\
    &-\mathbf{1}\left(\overline{\Delta}_t\ge0\right)
    \Pb\left(D_{i,t}>D_{i,1}\middle|Z_i=1\right)
    \frac{\overline{\Delta}_t}{FS_1}
    -\mathbf{1}\left(\overline{\Delta}_t<0\right)
    \max\left\{FS_t-FS_1,0\right\}
    \frac{\overline{\Delta}_t}{FS_1}
\end{split}
\end{equation*}
and an upper bound is given by
\begin{equation*}
\begin{split}
    \frac{RF_t}{FS_1}
    &+\mathbf{1}\left(\overline{\Delta}_t\ge0\right)
    \Pb\left(D_{i,t}>D_{i,1}\middle|Z_i=0\right)
    \frac{\overline{\Delta}_t}{FS_1}
    +\mathbf{1}\left(\overline{\Delta}_t<0\right)
    \max\left\{FS_1-FS_t,0\right\}
    \frac{\overline{\Delta}_t}{FS_1}\\
    &-\mathbf{1}\left(\underline{\Delta}_t<0\right)
    \Pb\left(D_{i,t}>D_{i,1}\middle|Z_i=1\right)
    \frac{\underline{\Delta}_t}{FS_1}
    -\mathbf{1}\left(\underline{\Delta}_t\ge0\right)
    \max\left\{FS_t-FS_1,0\right\}
    \frac{\underline{\Delta}_t}{FS_1}.
\end{split}
\end{equation*}
\end{propositionp}

\end{document}

%% file: figure.tex
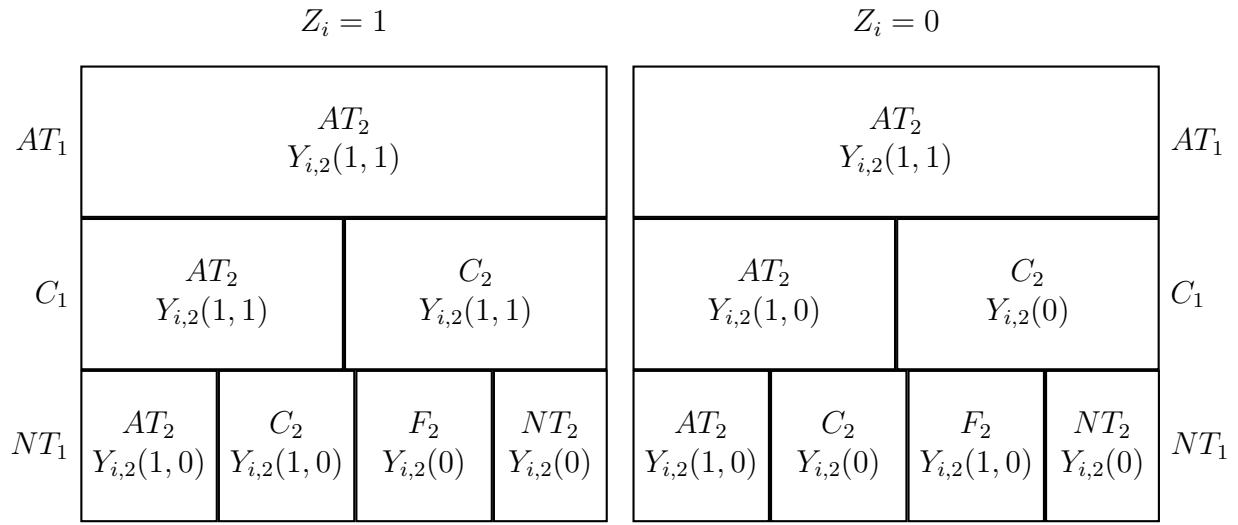
\begin{figure}[h]
\caption{Latent groups and potential outcomes when $Z_i=1$ and when $Z_i=0$.}
    \begin{tikzpicture}[mybox/.style={minimum width=6.98237cm,draw,thick,align=center,minimum height=2cm}, mybox2/.style={minimum width=1.8cm,draw,thick,align=center,minimum height=2cm}, gravity/.style={minimum width=2cm,align=center,minimum height=2cm}, mybox3/.style={minimum width=3.4765cm,draw,thick,align=center,minimum height=2cm}, mybox4/.style={minimum width=1.5cm,draw,thick,align=center,minimum height=2cm}]
\node[left, mybox,label=left:$AT_1$] (Z1AT1) {$AT_2$\\$Y_{i,2}(1,1)$};
\node[above=0.25cm of Z1AT1] (Z1) {$Z_i\eq1$};
\node[below=0.0065cm of Z1AT1, gravity] (gravity4) {};
\node[left=-1.007cm of gravity4, mybox3, label=left:$C_1$] (Z1C1AT2) {$AT_2$\\$Y_{i,2}(1,1)$};
\node[right=0cm of Z1C1AT2, mybox3] (Z1C1C2) {$C_2$\\$Y_{i,2}(1,1)$}; 
\node[below=0cm of Z1C1AT2, gravity] (gravity1) {};
\node[left=-1.0825cm of gravity1, mybox2, label=left:$NT_1$] (Z1NT1AT2) {$AT_2$\\$Y_{i,2}(1,0)$};
\node[right=0cm of Z1NT1AT2, mybox2] (Z1NT1C2) {$C_2$\\$Y_{i,2}(1,0)$};
\node[right=0cm of Z1NT1C2, mybox2] (Z1NT1F2) {$F_2$\\$Y_{i,2}(0)$};
\node[right=0cm of Z1NT1F2, mybox4] (Z1NT1NT2) {$NT_2$\\$Y_{i,2}(0)$};
%
\node[right=0.33cm of Z1AT1, mybox,label=right:$AT_1$] (Z0AT1) {$AT_2$\\$Y_{i,2}(1,1)$};
\node[above=0.25cm of Z0AT1] (Z0) {$Z_i\eq0$};
\node[below=0.0065cm of Z0AT1, gravity] (gravity2) {};
\node[left=-1.007cm of gravity2, mybox3] (Z0C1AT2) {$AT_2$\\$Y_{i,2}(1,0)$};
\node[right=0cm of Z0C1AT2, mybox3, label=right:$C_1$] (Z0C1C2) {$C_2$\\$Y_{i,2}(0)$};
\node[below=0cm of Z0C1AT2, gravity] (gravity3) {};
\node[left=-1.0825cm of gravity3, mybox2] (Z0NT1AT2) {$AT_2$\\$Y_{i,2}(1,0)$};
\node[right=0cm of Z0NT1AT2, mybox2] (Z0NT1C2) {$C_2$\\$Y_{i,2}(0)$};
\node[right=0cm of Z0NT1C2, mybox2] (Z0NT1F2) {$F_2$\\$Y_{i,2}(1, 0)$};
\node[right=0cm of Z0NT1F2, mybox4, label=right:$NT_1$] (Z0NT1NT2) {$NT_2$\\$Y_{i,2}(0)$};
\end{tikzpicture}
\label{figsecondperiod}
\end{figure}